\documentclass{amsart}

\usepackage{amssymb,amsmath,tikz}
\usetikzlibrary{positioning}
\usetikzlibrary{calc}
\usetikzlibrary{decorations.markings}

\newcommand{\id}{\mathrm{id}}

\newcommand{\C}{\mathbb{C}}
\newcommand{\R}{\mathbb{R}}

\newcommand{\T}{\mathbb{T}}
\newtheorem{theorem}{Theorem}

\theoremstyle{example}

\title{On realizations of Pachner moves in 4D}
\author{Rinat Kashaev}
\address{Section de math\'ematiques, Universit\'e de Gen\`eve,
2-4 rue du Li\`evre, 1211 Gen\`eve 4, Suisse\\}
\email{rinat.kashaev@unige.ch}
\date{April 8, 2015}

\thanks{Supported in part by Swiss National Science Foundation}
\begin{document}
\begin{abstract} 
The combinatorial structure of Pachner moves in four dimensions is analyzed in the case of a distinguished move of the type $(3,3)$ and few examples of solutions are reviewed. In particular, solutions associated to Pontryagin self-dual locally compact abelian groups are characterized with remarkable symmetry properties which, in the case of finite abelian groups, give rise to a simple model of combinatorial TQFT with corners in four dimensions. 
\end{abstract}
\maketitle
\section{Introduction}
The famous theorem of Pachner \cite{MR1095161} states that arbitrary triangulations of one and the same  piecewise linear (PL) manifold can be related by a finite sequence  of elementary moves of finitely many types known as  \emph{Pachner or bi-stellar moves} \cite{MR1734414}. In analogy with the Reidemeister theorem in knot theory, this result gives a combinatorial framework for constructing invariants of  PL-manifolds as well as PL topological quantum field theories (TQFT) with corners, provided one realizes the Pachner moves algebraically. In three dimensions, a first attempt  of using this scheme has been undertaken in the Regge--Ponzano model~\cite{PonzanoRegge1968}, where the Pachner moves are realized algebraically in terms of the $6j$-symbols within the quantum theory of the angular momentum. Subsequent developments have resulted in  the Turaev--Viro TQFT model~\cite{MR1191386} and its generalizations based on the theory of linear tensor categories~\cite{MR1292673}. In four dimensions, the problem is less developed than in lower dimensions, with few existing realizations \cite{MR1273569, MR1295461, MR1706684, MR1931150, MR3116189}. 

In this paper, following \cite{Kashaev2014(2)}, we review a specific approach to the algebraic realization problem of the Pachner moves in four dimensions. In particular, we discuss distribution valued solutions associated to Pontryagin self-dual locally compact abelian (LCA) groups which are characterized by a specific symmetry property which we call P-symmetry. As shown in \cite{Kashaev2014(2)}, in the case of finite abelian groups, these solutions give rise to a very simple model of TQFT in four dimensions. 

The paper is organized as follows. In Section~\ref{sec1}, we recall the definitions of topolological and combinatorial simplexes and the Pachner moves. In particular, we single out the unique distinguished Pachner move in each dimension. In Section~\ref{sec2}, we first discuss in more detail  the distinguished Pachner move in dimension four, introduce the string diagrammatic notation, and discuss few different interpretations, in particular, in terms of families of Yang--Baxter relations. Then, we describe particular realizations of algebraic, set theoretical and functional analytical natures. In the final Section~\ref{sec3}, we introduce the notion of a P-symmetric vector and prove Theorem~\ref{thm}, which generalizes the solution associated with finite cyclic groups used in \cite{Kashaev2014(2)} for construction of a particularly simple TQFT with corners in four dimensions.

\section{Simplexes and Pachner moves}\label{sec1}

\subsection{Standard topological simplexes}
Recall that the \emph{standard topological $n$-dimensional simplex} is a topological $n$-dimensional ball realized as the following subset of the $(n+1)$-dimensional unit cube:
\begin{equation}
\Delta^n=\left\{\left.(t_0,t_1,\dots,t_n)\in [0,1]^{n+1}\right \vert\ t_0+t_1+\dots+t_n=1\right\},
\end{equation}
so that its boundary 
\(
\partial \Delta^n,
\)
homeomorphic to the $(n-1)$-dimensional sphere, is composed  of $n+1$ homeomorphic images of the standard $(n-1)$-dimensional simplex through the \emph{face maps} 
\begin{multline}\label{fm}
f_i\colon\Delta^{n-1}\to \Delta^n,\quad (t_0,t_1,\dots,t_{n-1})\mapsto (t_0,\dots,t_{i-1},0,t_i,\dots,t_{n-1}), 
\\
 i\in
 [n]\equiv\{0,1,\dots,n\},
\end{multline}
which satisfy the relations
\begin{equation}\label{cfm}
f_i\circ f_j=f_{j+1}\circ f_i\ \text{ if }\  i\le j.
\end{equation}
The following picture illustrates the standard 2-dimensional simplex:
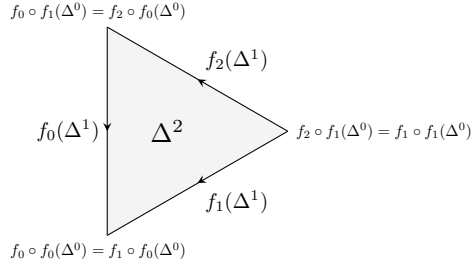
\[
\begin{tikzpicture}[scale=1.6,baseline=-3,decoration={
markings,
mark=at position .5 with {\arrow{stealth};}}]
\filldraw[color=gray!9] (0:1cm)--(120:1 cm)--(-120:1cm)--cycle;
\draw[auto,postaction={decorate}] (0:1cm) to node[scale=.8,swap] {$f_2(\Delta^1)$}(120:1 cm);
\draw[auto,postaction={decorate}] (120:1 cm) to node[scale=.8,swap] {$f_0(\Delta^1)$}(-120:1 cm);
\draw[auto,postaction={decorate}] (0:1 cm)to node[scale=.8] {$f_1(\Delta^1)$}(-120:1cm);
\node at (0:1.8cm) [scale=.6]  {$f_2\circ f_1(\Delta^0)=f_1\circ f_1(\Delta^0)$} ;
\node at (120:1.15cm)  [scale=.6]  {$f_0\circ f_1(\Delta^0)=f_2\circ f_0(\Delta^0)$} ;
\node at (-120:1.15cm)  [scale=.6]   {$f_0\circ f_0(\Delta^0)=f_1\circ f_0(\Delta^0)$} ;
\node at (0,0)  {$\Delta^2$} ;
 \end{tikzpicture}
 \]
where the arrows correspond to the standard orientation of $\Delta^1$. 

\subsection{Pachner moves}
Choose a splitting
\(
[n]=I\sqcup J
\)
with non-empty $I$ and $J$ of cardinality $p=|I|$ and $q=|J|$ so that $p+q=n+1$. The subsets 
\begin{equation}
\partial_I\Delta^n\equiv\cup_{i\in I}f_i(\Delta^{n-1})\quad\text{and}\quad \partial_J\Delta^n\equiv\cup_{j\in J}f_j(\Delta^{n-1})
\end{equation}
are topological $(n-1)$-dimensional balls sharing a common boundary (homeomorphic to the $(n-2)$-dimensional sphere $S^{n-2}$).
A \emph{Pachner move} of type $(p,q)$ is the operation of replacing a homeomorphic image of $\partial_I\Delta^n$ by  $\partial_J\Delta^n$ (a homeomorphic image thereof) in a triangulated $(n-1)$-dimensional manifold. 

 The Pachner move corresponding to the 
specific splitting $[n]=I_0\sqcup J_0$ with 
\begin{equation}
 I_0=\left\{0,2,\dots,2\left\lfloor \frac n2\right\rfloor\right\}\quad \text{and}\quad 
  J_0=\left\{1,3,\dots,2\left\lceil \frac n2\right\rceil-1\right\}
\end{equation}
 will be called the \emph{distinguished  Pachner move}. Here we use the floor $\lfloor x\rfloor$ (the largest integer not greater than $x$) and the ceiling $\lceil x\rceil $ (the smallest integer not less than $x$) functions. The corresponding partition of $n+1$ is given by the formula
\begin{equation}
 (p_0,q_0)=(|I_0|,|J_0|)=\left(\left\lfloor \frac n2\right\rfloor+1,\left\lceil \frac n2\right\rceil\right).
\end{equation}
The following table gives the values for small $n=d+1$:
\begin{equation}
\begin{array}{c|ccccccc}
d&1&2&3&4&5&6&7\\
\hline
(p_0,q_0)  & (2,1)  &(2,2)  &(3,2)&(3,3)&(4,3)&(4,4) & (5,4) \\
\end{array}
\end{equation}
The following pictures illustrate few Pachner moves in small dimensions.
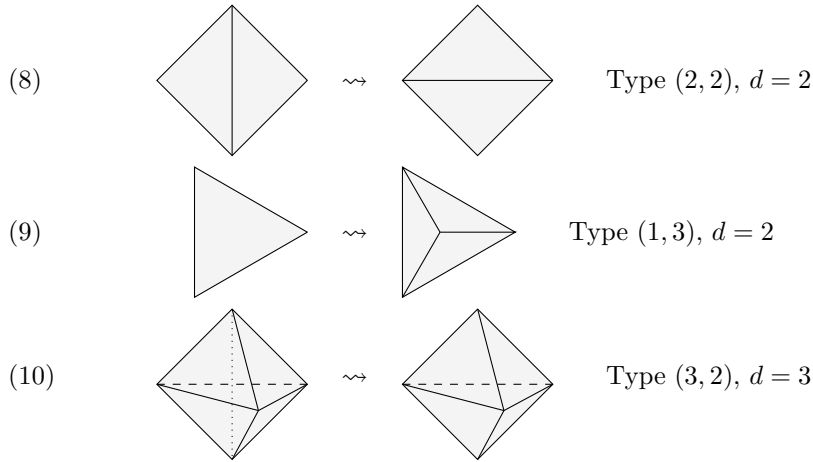
\begin{align}
 \begin{tikzpicture}[scale=1,baseline=-3]
\filldraw[color=gray!9] (0:1cm)--(90:1 cm)--(180:1cm)--(-90:1cm)--cycle;
\draw[auto] (90:1 cm)--(-90:1cm) ;
\draw (0:1cm)--(90:1 cm)--(180:1 cm)--(-90:1cm)--(0:1cm);
 \end{tikzpicture}\quad &\rightsquigarrow\quad
 \begin{tikzpicture}[scale=1,baseline=-3]
\filldraw[color=gray!9] (0:1cm)--(90:1 cm)--(180:1cm)--(-90:1cm)--cycle;
\draw[auto](180:1 cm)--(0:1cm) ;
\draw (0:1cm)--(90:1 cm)--(180:1 cm)--(-90:1cm)--(0:1cm);
 \end{tikzpicture}
 \qquad \text{Type $(2,2)$, $d=2$}\\
 \begin{tikzpicture}[scale=1,baseline=-3]
\filldraw[color=gray!9] (0:1cm)--(120:1 cm)--(-120:1cm)--cycle;
\draw (0:1cm)--(120:1 cm)--(-120:1 cm)--(0:1cm);
 \end{tikzpicture}\quad&\rightsquigarrow\quad
 \begin{tikzpicture}[scale=1,baseline=-3]
\filldraw[color=gray!9] (0:1cm)--(120:1 cm)--(-120:1cm)--cycle;
\draw (0:1cm)--(120:1 cm)--(-120:1 cm)--(0:1cm)--(0,0)--(120:1cm);
\draw (0,0)--(-120:1 cm);
 \end{tikzpicture}
 \qquad \text{Type $(1,3)$, $d=2$}\\
\begin{tikzpicture}[baseline,scale=0.5]
\draw[fill=gray!9] (2,0)--(0,2)--(-2,0)--(0,-2)--cycle;
\draw (0,2)--(.7,-.7)--(0,-2);
\draw (-2,0)--(.7,-.7)--(2,0);
\draw[dashed] (-2,0)--(2,0);
\draw[dotted] (0,2)--(0,-2);
\end{tikzpicture}
\quad&\rightsquigarrow\quad
\begin{tikzpicture}[baseline,scale=0.5]
\draw[fill=gray!9] (2,0)--(0,2)--(-2,0)--(0,-2)--cycle;
\draw (0,2)--(.7,-.7)--(0,-2);
\draw (-2,0)--(.7,-.7)--(2,0);
\draw[dashed] (-2,0)--(2,0);
\end{tikzpicture}
\qquad \text{Type $(3,2)$, $d=3$}
\end{align}
 The Pachner's theorem \cite{MR1095161}, which states that
 two triangulated PL-manifolds are PL-homeomorphic if and only if they are related by a finite sequence of Pachner moves, allows to construct PL-invariants provided one finds algebraic realizations of the Pachner moves. A typical example are the Turaev--Viro invariants of 3-manifolds \cite{MR1191386} which are based on the theory of 6j-symbols of the representation theory of Hopf algebras giving algebraic realizations of the Pachner moves of the types $(p,5-p)$, $p\in\{1,2,3,4\}$, which are also known as \emph{Matveev--Piergallini moves} in the dual language of special spines of 3-manifolds \cite{MR925096,MR1997069,MR958750}.
 
  \subsection{Combinatorial simplexes}
For algebraic realizations of the Pachner moves, we will find it useful to work with combinatorial simplexes.
By definition,
 a \emph{combinatorial $n$-simplex} $\Delta S$ is an abstract simplicial complex given by the power set of a linearly ordered set $S$ of cardinality $n+1$.
 The \emph{standard combinatorial $n$-simplex} $\Delta[n]$ corresponds to the set $[n]\equiv \{0,1,\dots,n\}$ with the natural linear order. It is clear that for any combinatorial $n$-simplex $\Delta S$ there is a unique bijection $s\colon [n]\to S$, $i\mapsto s_i$,  preserving the linear order.
The \emph{facets} of an $n$-simplex $\Delta S$ are combinatorial $(n-1)$-simplexes  $\partial_i \Delta S$, $i\in [n]$, corresponding to the subsets
\(
S\setminus \{s_i\}
\)
with the natural induced linear order:
\begin{equation}
\partial_i \Delta S= \Delta (S\setminus \{s_i\}).
\end{equation}
Pictorially, one can still think in terms of the topological simplexes, for example, the 2-simplex $\Delta [2]$
can be drawn as follows
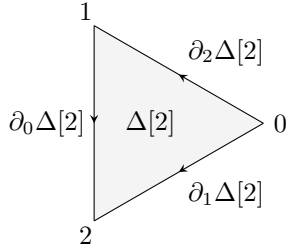
\[
 \begin{tikzpicture}[scale=1.5,baseline=-3,decoration={
markings,
mark=at position .5 with {\arrow{stealth};}}]
\filldraw[color=gray!9] (0:1cm)--(120:1 cm)--(-120:1cm)--cycle;
\draw[auto,postaction={decorate}] (0:1cm) to node[swap] {$\partial_2\Delta[2]$}(120:1 cm);
\draw[auto,postaction={decorate}] (120:1 cm) to node[swap] {$\partial_0\Delta[2]$}(-120:1 cm);
\draw[auto,postaction={decorate}] (0:1 cm)to node {$\partial_1\Delta[2]$}(-120:1cm);
\node at (0:1.15cm)  {$0$} ;
\node at (120:1.15cm)  {$1$} ;
\node at (-120:1.15cm)  {$2$} ;
\node at (0,0)  {$\Delta[2]$} ;
 \end{tikzpicture}
 \]
but one should keep in mind that, compared to the face maps \eqref{fm}, the operations of taking the boundary components are composed in the opposite order, so that the identity \eqref{cfm} should be replaced by
\begin{equation}
\partial_j\partial_i=\partial_i\partial_{j+1}\ \text{ if }\  i\le j.
\end{equation}
\section{Algebraic realizations}\label{sec2}
A ``good'' algebraic realization of the Pachner moves in $d$ dimensions should result in a TQFT (with corners) in $d$ dimensions. That means that among other things one should be able, for example,  to associate
  a complex vector space $V(s)$ to each $(d-1)$-simplex $s$, and
 a vector $V(p)\in V(\partial_0p)\otimes\dots\otimes V(\partial_dp)$ to each $d$-simplex $p$ in such a way that if $s^*$ is a $(d-1)$-simplex $s$ with opposite (induced) orientation then the associated vector space should be the dual vector space $V(s)^*=V(s^*)$, and  functorial properties with respect to the gluing operations should also be satisfied.  In particular, for the distinguished Pachner moves, one should have an equality of the form
 \begin{equation}
\left(\otimes_{k<l}\operatorname{Ev}_{2k,2l}\right)\left(\otimes_{i}V(\partial_{2i}u)\right)
=\left(\otimes_{k<l}\operatorname{Ev}_{2k+1,2l+1}\right)\left(\otimes_{j}V(\partial_{2j+1}u)\right)
\end{equation}
where $u$ is a $(d+1)$-simplex and  
\begin{equation}
\operatorname{Ev}_{i,j}\colon V(\partial_i\partial_ju)^*\otimes V(\partial_i\partial_ju)\to \C
\end{equation}
is the operation of contracting between elements of dual vectors spaces. In this way, one arrives at a multilinear algebraic relation on the vectors associated to the $d$-simplexes. The following table gives a list of algebraic structures which enable to realize at least the distinguished Pachner moves in dimensions up to 3:
\begin{equation*}
\begin{tabular}{c||c|c|c}
$d$&1&2&3\\
\hline
algebraic structure&projectors&algebras& bi-algebras
\end{tabular}
\end{equation*}
For example, in the case $d=3$, the distinguished Pachner move of the type $(3,2)$ can be realized by a linear map 
 $S\colon V\otimes V\to V\otimes V$ satisfying the equality
 \begin{equation}\label{pe}
 S_{1,2} S_{1,3} S_{2,3}= S_{2,3} S_{1,2}
\end{equation}
where  $V$ is a vector space and $S_{i,j}$ is an element of the endomorphism algebra of the vector space $V^{\otimes3}$ acting as $S$ in the $i$-th and $j$-th components and identically in the remaining component, e.g. $S_{1,2}=S\otimes\operatorname{id}_V$. Namely, for any oriented 2-simplex (triangle) $t$ we assign the vector space $V(t)=V$ and for any oriented 3-simplex (tetrahedron) $s$ the vector 
\begin{multline}
V(s)=S\in V(\partial_0s)\otimes V(\partial_1s)\otimes V(\partial_2s)\otimes V(\partial_3s)\\=V\otimes V^*\otimes V\otimes V^*=\operatorname{End}(V\otimes V)
\end{multline}
where we use the fact that the induced orientations on the even facets of a simplex are opposite to those on the odd facets.
If $V$ is provided with the structure of a bi-algebra with the multiplication $\nabla_V\colon V\otimes V\to V$ and the co-multiplication $\Delta_V\colon V\to V\otimes V$, then the element
 \begin{equation}
 S=(\operatorname{id}_V\otimes \nabla_V)(\Delta_V\otimes \operatorname{id}_V)\colon V\otimes V\to V\otimes V
\end{equation}
is easily seen to satisfy relation~\eqref{pe}.

\subsection{The case $d=4$} Let $p$  be  an oriented 4-simplex also called \emph{pentachoron}. We will find it convenient to represent a pentachoron diagrammatically as a five valent vertex
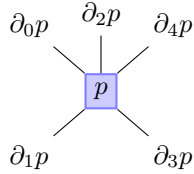
\[
  \begin{tikzpicture}[node distance=.5cm,baseline=(x.base),
  hvector/.style={draw=blue!50,fill=blue!20,thick}]
\node (x) [hvector]{$p$};
\node (i) [above left = of x ]{$\partial_0 p$};
\node (j) [above=of x]{$\partial_2 p$};
\node (k) [above right =of x]{$\partial_4 p$};
\node (l) [below left =of x]{$\partial_1 p$};
\node (m) [below right =of x]{$\partial_3 p$};
\draw (x.north west)-- (i)(x)--(j)(x.north east)--(k)(x.south west)--(l)(x.south east)--(m);
 \end{tikzpicture}
 \]
 which allows to represent the 
distinguished Pachner move of the type $(3,3)$ as the following (string) diagrammatic equality 
 
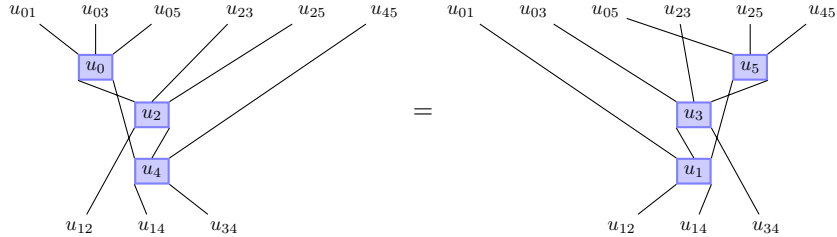
\[
  \begin{tikzpicture}[node distance=.4cm,baseline=(y.base),
  hvector/.style={draw=blue!50,fill=blue!20,thick}]
\node (i) [scale=.75]{$u_{01}$};
\node (l) [scale=.75,right=of i]{$u_{03}$};
\node (m) [scale=.75,right =of l]{$u_{05}$};
\node (j) [scale=.75,right =of m]{$u_{23}$};
\node (n) [scale=.75,right =of j]{$u_{25}$};
\node (k) [scale=.75,right =of n]{$u_{45}$};
\node (x) [scale=.75,hvector,below =of l]{$u_0$};
\node (y) [scale=.75,hvector,below right=of x]{$u_2$};
\node (z) [scale=.75,hvector,below=of y]{$u_4$};
\node (q) [scale=.75,below =of z ]{$u_{14}$};
\node (r) [scale=.75,right =of q]{$u_{34}$};
\node (p) [scale=.75,left =of q]{$u_{12}$};
\draw (x.north west)-- (i)(x.north)--(l)(x.north east)--(m)(x.south west)--(y.north west)(x.south east)--(z.north west)(y.north)--(j)(y.north east)--(n)(y.south west)--(p)(y.south east)--(z.north)(z.north east)--(k)
(z.south west)--(q)(z.south east)--(r);
 \end{tikzpicture}
 =
  \begin{tikzpicture}[node distance=.4cm,baseline=(y.base),
  hvector/.style={draw=blue!50,fill=blue!20,thick}]
\node (i) [scale=.75]{$u_{01}$};
\node (l) [scale=.75,right=of i]{$u_{03}$};
\node (m) [scale=.75,right =of l]{$u_{05}$};
\node (j) [scale=.75,right =of m]{$u_{23}$};
\node (n) [scale=.75,right =of j]{$u_{25}$};
\node (k) [scale=.75,right =of n]{$u_{45}$};
\node (x) [scale=.75,hvector,below =of n]{$u_5$};
\node (y) [scale=.75,hvector,below left=of x]{$u_3$};
\node (z) [scale=.75,hvector,below=of y]{$u_1$};
\node (q) [scale=.75,below =of z ]{$u_{14}$};
\node (r) [scale=.75,right =of q]{$u_{34}$};
\node (p) [scale=.75,left =of q]{$u_{12}$};
\draw (x.north west)-- (m)(x.north)--(n)(x.north east)--(k)(x.south west)--(z.north east)(x.south east)--(y.north east)(y.north)--(j)(y.north west)--(l)(y.south east)--(r)(y.south west)--(z.north)(z.north west)--(i)
(z.south west)--(p)(z.south east)--(q);
 \end{tikzpicture}
 \]
 where $u$ is a 5-simplex with the notation
 \[
u_i\equiv\partial_iu,\quad u_{ij}\equiv\partial_i\partial_ju,
\]
and the  internal edges correspond to internal shared tetrahedra. Specifically, the internal edge connecting the pentachora $u_i$ and $u_j$ with $i<j$ correspond to the common internal tetrahedron $u_{ij}$ shared between those pentachora. 

Using the fact that the induced  orientations on even facets are opposite to those on odd facets, the   vector $V(p)\in V(\partial_0p)\otimes\dots\otimes V(\partial_4p)$ can be thought of as a linear map
\begin{equation*}
V(p)\colon V(\partial_1p^*)\otimes V(\partial_3p^*)\to V(\partial_0p)\otimes V(\partial_2p)\otimes V(\partial_4p)
\end{equation*}
with a similar string diagrammatical notation as for an abstract pentachoron
 
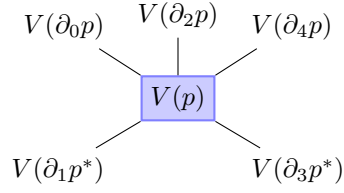
\[
  \begin{tikzpicture}[node distance=.5cm,baseline=(x.base),
  hvector/.style={draw=blue!50,fill=blue!20,thick}]
\node (x) [hvector]{$V(p)$};
\node (i) [above left = of x ]{$V(\partial_0 p)$};
\node (j) [above=of x]{$V(\partial_2 p)$};
\node (k) [above right =of x]{$V(\partial_4 p)$};
\node (l) [below left =of x]{$V(\partial_1 p^*)$};
\node (m) [below right =of x]{$V(\partial_3 p^*)$};
\draw (x.north west)-- (i)(x)--(j)(x.north east)--(k)(x.south west)--(l)(x.south east)--(m);
 \end{tikzpicture}
 \]
 In the simplest case of a constant realization 
  \begin{equation}
 V(p)=Q\ \equiv\begin{tikzpicture}[node distance=.3cm,baseline=(x.base),
  hvector/.style={draw=blue!50,fill=blue!20,thick}]
\node (x) [hvector]{$Q$};
\node (i) [above left = of x ]{
};
\node (j) [above=of x]{
};
\node (k) [above right =of x]{
};
\node (l) [below left =of x]{
};
\node (m) [below right =of x]{
};
\draw (x.north west)-- (i)(x)--(j)(x.north east)--(k)(x.south west)--(l)(x.south east)--(m);
 \end{tikzpicture}\colon V\otimes V\to V\otimes V\otimes V 
  \end{equation}
 we arrive at a graphical equation to be called \emph{Pachner (3,3)-relation}
  \begin{equation}\label{p33rg}
 \begin{tikzpicture}[node distance=.4cm,baseline=(y.base),
  hvector/.style={draw=blue!50,fill=blue!20,thick}]
\node (i)[scale=.75]{};
\node (l) [scale=.75,right=of i]{};
\node (m) [scale=.75,right =of l]{};
\node (j) [scale=.75,right =of m]{};
\node (n) [scale=.75,right =of j]{};
\node (k) [scale=.75,right =of n]{};
\node (x) [scale=.75,hvector,below =of l]{$Q$};
\node (y) [scale=.75,hvector,below right=of x]{$Q$};
\node (z) [scale=.75,hvector,below=of y]{$Q$};
\node (q) [scale=.75,below =of z ]{};
\node (r) [scale=.75,right =of q]{};
\node (p) [scale=.75,left =of q]{};
\draw (x.north west)-- (i)(x.north)--(l)(x.north east)--(m)(x.south west)--(y.north west)(x.south east)--(z.north west)(y.north)--(j)(y.north east)--(n)(y.south west)--(p)(y.south east)--(z.north)(z.north east)--(k)
(z.south west)--(q)(z.south east)--(r);
 \end{tikzpicture}
 =
  \begin{tikzpicture}[node distance=.4cm,baseline=(y.base),
  hvector/.style={draw=blue!50,fill=blue!20,thick}]
\node (i)[scale=.75]{};
\node (l) [scale=.75,right=of i]{};
\node (m) [scale=.75,right =of l]{};
\node (j) [scale=.75,right =of m]{};
\node (n) [scale=.75,right =of j]{};
\node (k) [scale=.75,right =of n]{};
\node (x) [scale=.75,hvector,below =of n]{$Q$};
\node (y) [scale=.75,hvector,below left=of x]{$Q$};
\node (z) [scale=.75,hvector,below=of y]{$Q$};
\node (q) [scale=.75,below =of z ]{};
\node (r) [scale=.75,right =of q]{};
\node (p) [scale=.75,left =of q]{};
\draw (x.north west)-- (m)(x.north)--(n)(x.north east)--(k)(x.south west)--(z.north east)(x.south east)--(y.north east)(y.north)--(j)(y.north west)--(l)(y.south east)--(r)(y.south west)--(z.north)(z.north west)--(i)
(z.south west)--(p)(z.south east)--(q);
 \end{tikzpicture}
\end{equation}
whose analytic form reads as
\begin{multline}\label{p33ra}
\left(Q\sigma\otimes\operatorname{id}_{V^{\otimes3}}\right)\left(\operatorname{id}_{V}\otimes Q\otimes \operatorname{id}_{V}\right)
\left(\sigma\otimes\operatorname{id}_{V^{\otimes2}}\right)\left(\operatorname{id}_{V}\otimes Q\right)\\
=\left(\operatorname{id}_{V^{\otimes2}}\otimes \sigma\otimes \operatorname{id}_{V^{\otimes2}}\right)
\left(\operatorname{id}_{V^{\otimes3}}\otimes Q\sigma\right)\left(\operatorname{id}_{V}\otimes Q\otimes \operatorname{id}_{V}\right)\left(\operatorname{id}_{V^{\otimes2}}\otimes \sigma\right)\left(Q\otimes \operatorname{id}_{V}\right)
\end{multline}
where $\sigma=\sigma_{V,V}$ is the permutation operator defined by 
\begin{equation}\label{perm}
\sigma_{X,Y}\colon X\otimes Y\to Y\otimes X,\ x\otimes y\mapsto y\otimes x,\quad \forall x\in X,\ \forall y\in Y.
\end{equation}
We  can also use the same  graphical notation
 \[
  Q^{i,j,k}_{l,m}\ \equiv \begin{tikzpicture}[node distance=.3cm,baseline=(x.base),
  hvector/.style={draw=blue!50,fill=blue!20,thick}]
\node (x) [hvector]{$Q$};
\node (i) [above left = of x ]{$i$};
\node (j) [above=of x]{$j$};
\node (k) [above right =of x]{$k$};
\node (l) [below left =of x]{$l$};
\node (m) [below right =of x]{$m$};
\draw (x.north west)-- (i)(x)--(j)(x.north east)--(k)(x.south west)--(l)(x.south east)--(m);
 \end{tikzpicture}
 \in \mathbb{C}
 \]
 for the \emph{matrix coefficients} 
 \[
 Q^{i,j,k}_{l,m}\equiv\langle Q, \omega^i\otimes e_l\otimes\omega^j\otimes e_m\otimes\omega^k\rangle=\langle \omega^i\otimes\omega^j\otimes\omega^k,Q(e_l\otimes e_m)\rangle
 \]
 associated with dual linear bases $\{e_i\}\subset V$ and $\{\omega^j\}\subset V^*$.
 In this way, we obtain a coordinate form of the Pachner (3,3)-relation given by a system of non-linear algebraic equations
 
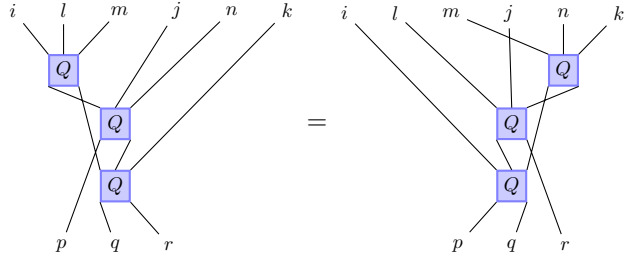
\[
  \begin{tikzpicture}[node distance=.4cm,baseline=(y.base),
  hvector/.style={draw=blue!50,fill=blue!20,thick}]
\node (i)[scale=.75]{$i$};
\node (l) [scale=.75,right=of i]{$l$};
\node (m) [scale=.75,right =of l]{$m$};
\node (j) [scale=.75,right =of m]{$j$};
\node (n) [scale=.75,right =of j]{$n$};
\node (k) [scale=.75,right =of n]{$k$};
\node (x) [scale=.75,hvector,below =of l]{$Q$};
\node (y) [scale=.75,hvector,below right=of x]{$Q$};
\node (z) [scale=.75,hvector,below=of y]{$Q$};
\node (q) [scale=.75,below =of z ]{$q$};
\node (r) [scale=.75,right =of q]{$r$};
\node (p) [scale=.75,left =of q]{$p$};
\draw (x.north west)-- (i)(x.north)--(l)(x.north east)--(m)(x.south west)--(y.north west)(x.south east)--(z.north west)(y.north)--(j)(y.north east)--(n)(y.south west)--(p)(y.south east)--(z.north)(z.north east)--(k)
(z.south west)--(q)(z.south east)--(r);
 \end{tikzpicture}
 =
  \begin{tikzpicture}[node distance=.4cm,baseline=(y.base),
  hvector/.style={draw=blue!50,fill=blue!20,thick}]
\node (i)[scale=.75]{$i$};
\node (l) [scale=.75,right=of i]{$l$};
\node (m) [scale=.75,right =of l]{$m$};
\node (j) [scale=.75,right =of m]{$j$};
\node (n) [scale=.75,right =of j]{$n$};
\node (k) [scale=.75,right =of n]{$k$};
\node (x) [scale=.75,hvector,below =of n]{$Q$};
\node (y) [scale=.75,hvector,below left=of x]{$Q$};
\node (z) [scale=.75,hvector,below=of y]{$Q$};
\node (q) [scale=.75,below =of z ]{$q$};
\node (r) [scale=.75,right =of q]{$r$};
\node (p) [scale=.75,left =of q]{$p$};
\draw (x.north west)-- (m)(x.north)--(n)(x.north east)--(k)(x.south west)--(z.north east)(x.south east)--(y.north east)(y.north)--(j)(y.north west)--(l)(y.south east)--(r)(y.south west)--(z.north)(z.north west)--(i)
(z.south west)--(p)(z.south east)--(q);
 \end{tikzpicture}
 \]
where summations over the implicit indices on the internal edges are assumed. The explicit analytic form of these equations reads as follows:
\begin{equation}\label{p33r}
\sum_{s,t,u}Q^{i,l,m}_{s,t}Q^{s,j,n}_{p,u}Q^{t,u,k}_{q,r}=\sum_{s,t,u}Q^{m,n,k}_{s,t}Q^{l,j,t}_{u,r}Q^{i,u,s}_{p,q}.
\end{equation}
\subsection{Interpretations}
Sticking with chosen dual linear bases $\{e_i\}\subset V$ and $\{\omega^j\}\subset V^*$, define elements $ X^i,Y^j,Z^k\in\operatorname{End}(V\otimes V)$ through the equalities
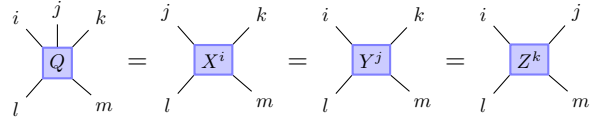
\[
 \begin{tikzpicture}[node distance=.3cm,baseline=(x.base),
  hvector/.style={draw=blue!50,fill=blue!20,thick}]
\node (x) [scale=.75,hvector]{$Q$};
\node (i) [scale=.75,above left = of x ]{$i$};
\node (j) [scale=.75,above=of x]{$j$};
\node (k) [scale=.75,above right =of x]{$k$};
\node (l) [scale=.75,below left =of x]{$l$};
\node (m) [scale=.75,below right =of x]{$m$};
\draw (x.north west)-- (i)(x)--(j)(x.north east)--(k)(x.south west)--(l)(x.south east)--(m);
 \end{tikzpicture}
 = \begin{tikzpicture}[node distance=.3cm,baseline=(x.base),
  hvector/.style={draw=blue!50,fill=blue!20,thick}]
\node (x) [scale=.75,hvector]{$X^i$};
\node (i) [scale=.75,above left = of x ]{$j$};
\node (k) [scale=.75,above right =of x]{$k$};
\node (l) [scale=.75,below left =of x]{$l$};
\node (m) [scale=.75,below right =of x]{$m$};
\draw (x.north west)-- (i)(x.north east)--(k)(x.south west)--(l)(x.south east)--(m);
 \end{tikzpicture}
 =
  \begin{tikzpicture}[node distance=.3cm,baseline=(x.base),
  hvector/.style={draw=blue!50,fill=blue!20,thick}]
\node (x) [scale=.75,hvector]{$Y^j$};
\node (i) [scale=.75,above left = of x ]{$i$};
\node (k) [scale=.75,above right =of x]{$k$};
\node (l) [scale=.75,below left =of x]{$l$};
\node (m) [scale=.75,below right =of x]{$m$};
\draw (x.north west)-- (i)(x.north east)--(k)(x.south west)--(l)(x.south east)--(m);
 \end{tikzpicture}
 =
  \begin{tikzpicture}[node distance=.3cm,baseline=(x.base),
  hvector/.style={draw=blue!50,fill=blue!20,thick}]
\node (x) [scale=.75,hvector]{$Z^k$};
\node (i) [scale=.75,above left = of x ]{$i$};
\node (k) [scale=.75,above right =of x]{$j$};
\node (l) [scale=.75,below left =of x]{$l$};
\node (m) [scale=.75,below right =of x]{$m$};
\draw (x.north west)-- (i)(x.north east)--(k)(x.south west)--(l)(x.south east)--(m);
 \end{tikzpicture}
\]
where we use the graphical notation for the matrix coefficents
\begin{equation}
 \begin{tikzpicture}[node distance=.3cm,baseline=(x.base),
  hvector/.style={draw=blue!50,fill=blue!20,thick}]
\node (x) [scale=.75,hvector]{$W$};
\node (i) [scale=.75,above left = of x ]{$i$};
\node (k) [scale=.75,above right =of x]{$j$};
\node (l) [scale=.75,below left =of x]{$k$};
\node (m) [scale=.75,below right =of x]{$l$};
\draw (x.north west)-- (i)(x.north east)--(k)(x.south west)--(l)(x.south east)--(m);
 \end{tikzpicture}
\equiv \ \langle \omega^i\otimes \omega^j,W (e_k\otimes e_l)\rangle,\quad \forall\  W\in \operatorname{End}(V\otimes V).
\end{equation}
Then we can write the coordinate form~\eqref{p33r} of the Pachner $(3,3)$-relation in equivalent forms as either
\begin{equation}\label{pe1}
\sum_{s,t}Q^{i,l,m}_{s,t}X^{s}_{1,2}X^{t}_{2,3}=X^{m}_{2,3}X^l_{1,3}X^{i}_{1,2}
\end{equation}
or 
\begin{equation}\label{pe2}
Z^{m}_{1,2}Z^{n}_{1,3}Z^{k}_{2,3}=\sum_{s,t}Q^{m,n,k}_{s,t}Z^{t}_{2,3}Z^{s}_{1,2}
\end{equation}
or else
\begin{equation}\label{ybe}
X^i_{1,2}Y^j_{1,3}Z^k_{2,3}=Z^k_{2,3}Y^j_{1,3}X^i_{1,2}
\end{equation}
where the lower indices carry the matrix meaning (see the explanation right after equation~\eqref{pe}) while the upper indices enumerate the elements of the dual linear bases of $V$ and $V^*$.
If equations~\eqref{pe1} and \eqref{pe2} can naturally be interpreted as categorified versions of the pentagon equation~\eqref{pe}, 
equations~\eqref{ybe} are easily recognized as a family of Yang--Baxter equations, and it is an interesting open question if one can produce nontrivial solutions to the Pachner $(3,3)$-relation by studying some known (families of) solutions of the Yang--Baxter equation.
\subsection{Solutions}
Let $V$ be a vector space together with a linear map  $\mu\colon V\otimes V\to V$ satisfying the associativity condition 
\begin{equation}
\mu(\mu\otimes\operatorname{id}_V)=\mu(\operatorname{id}_V\otimes\mu)
\end{equation}
and  two linear maps $\lambda,\rho\colon V\to V\otimes V$ satisfying the co-associativity condition
\begin{equation}
(\operatorname{id}_V\otimes x)x=(x\otimes\operatorname{id}_V)x,\quad x\in\{\lambda,\rho\},
\end{equation}
 and suppose that all three maps are pairwise mutually compatible in the sense that any of them is a morphism of the algebraic structures defined by two others, i.e. the following equations are satisfied:
 \begin{equation}
\mu x=(x\otimes x)(\operatorname{id}_V\otimes\sigma\otimes\operatorname{id}_V)(\mu\otimes\mu),\quad x\in\{\lambda,\rho\},
\end{equation}
the compatibility of $\mu$ with $\lambda$ and $\rho$, and
\begin{equation}
\lambda(\rho\otimes\rho)=\rho(\lambda\otimes\lambda)(\operatorname{id}_V\otimes\sigma\otimes\operatorname{id}_V),
\end{equation}
the compatibility of $\lambda$ with $\rho$, where $\sigma=\sigma_{V,V}$ is the permutation operator defined in \eqref{perm}.
  Then the map
\begin{equation}
 Q=(\id_V\otimes\mu\otimes\id_V)(\lambda\otimes\rho)\colon V\otimes V\to V\otimes V\otimes V
\end{equation}
satisfies  the Pachner (3,3)-relation~\eqref{p33rg}, \eqref{p33ra}.
In particular, one can take a vector space $V$ provided with the structure of a co-commutative bi-algebra,
with the multiplication $\nabla_V\colon V\otimes V\to V$ and the co-multiplication $\Delta_V\colon V\to V\otimes V$, and choose $\mu=\nabla_V$ and $\lambda=\rho=\Delta_V$.

By putting oneself into the set theoretical framework,  instead of the category of vector spaces, we identify a set theoretical solution 
\begin{equation}
 Q\colon I^2\to I^3,\quad (x,y)\mapsto \left(\frac{x-xy}{1-xy},xy,\frac{y-xy}{1-xy}\right)
\end{equation}
 where $I=]0,1[\subset\R$ is the open unit interval. The idea of why this map is a solution comes from the fact that exactly those rational expressions appear as arguments of the dilogarithm functions in the five term Rogers identity.

Finally, for any abelian group $A$ and  a bi-character $\chi\colon A\times A\to\mathbb{C}$, we have a solution in the form of a linear operator between the vector spaces of complex valued functions on $A^2$ and $A^3$:
\begin{equation}
Q\colon \C^{A^2}\to \C^{A^3},\quad (Qf)(x,y,z)=\chi(x,z)f(x+y,y+z).
\end{equation}
One can define the matrix coefficients for that solution if $A$ is a locally compact abelian (abbreviated as LCA) group so that the matrix coefficients make sense as tempered distributions over $A^5$:  
 \begin{equation}\label{lcas}
 Q^{x,y,z}_{u,v}=D(x,u,y,v,z)\equiv\chi(x,z)\delta(x-u+y)\delta(y-v+z),\quad \forall (x,y,z,u,v)\in A^5,
\end{equation}
where $\delta(x)$ is Dirac's delta distribution over $A$ defined by
\begin{equation}
\int_A \delta(x)f(x)\operatorname{d}\! x=f(0)
\end{equation}
with a chosen Haar measure $\operatorname{d}\! x$ on $A$ and 
 any Schwartz--Bruhat test function $f\colon A\to \C$. One can easily check that equations~\eqref{p33r} are satisfied by \eqref{lcas}, provided the summations are interpreted as integrations with respect to the Haar measure. 
 \section{P-symmetries}\label{sec3}
 
 Let $X$ be a complex vector space and $Y\subset X^*$ a subspace of the corresponding dual vector space. We say that a linear map $A\colon Y\to X$ is \emph{symmetric} if the associated bilinear form on $Y$ is symmetric, i.e.
 \begin{equation}
\langle v,Aw\rangle =\langle w,Av\rangle,\quad \forall v,w\in Y.
\end{equation}
We also say that a vector 
$
Q\in Y\otimes X\otimes Y\otimes X\otimes Y
$
is \emph{P-symmetric} if there exist three symmetric linear isomorphisms
\begin{equation}
L,M,R \colon Y\to X
\end{equation}
such that the following equalities are satisfied:
\begin{multline}\label{eq:sym-rel}
(\sigma_{Y,X}\otimes L\otimes L^{-1}\otimes L)Q=
(L\otimes \sigma_{X,Y} \otimes M^{-1}\otimes M)Q\\
=(M\otimes M^{-1}\otimes  \sigma_{Y,X}\otimes R) Q
=(R\otimes R^{-1}\otimes R\otimes   \sigma_{X,Y} )Q
\end{multline}
where the permutation operators $\sigma_{X,Y}$ are defined in \eqref{perm}. If a P-symmetric vector $Q$ satisfies the Pachner (3,3)-relation\footnote{through the interpretation as a linear map $Q\colon Y^{\otimes2}\to Y^{\otimes3}$}, thus realizing algebraically  the distinguished Pachner move of the type (3,3), then it can be used for realization of all other Pachner moves of the same type. Likewise, for realization of any of the Pachner moves of the type $(q,6-q)$ with $q\in \{1,2,4,5\}$ it suffices to  realize  one Pachner move of the type $(1,5)$ and one of the type $(2,4)$ with a P-symmetric vector, see \cite{Kashaev2014(2)} for an example of P-symmetric solutions in a particular case of finite cyclic groups.
\subsection{P-symmetric solutions over self-dual LCA groups} Here we give a functional analytic example of P-symmetric solutions of the Pachner (3,3)-relation which generalizes the solution described in  \cite{Kashaev2014(2)}.

Let $\T$ denote the complex circle group. For any LCA  group $A$ we denote by  $\hat A$ its Pontryagin dual and by $\langle \xi;x\rangle$ the (duality) pairing between  $\xi \in\hat A$ and $x\in A$. Let $A$ be a self-dual LCA group. That means that $\hat A$ is isomorphic to $A$. We assume that one can choose an isomorphism $f\colon A\to \hat A$ so that there exists a function $g\colon A\to \T$ which is symmetric $g(x)=g(-x)$ and trivializes the group cocycle $\langle f(x);y\rangle$, i.e.
\begin{equation}
\langle f(x);y\rangle={g(x)g(y)\over g(x+y)}.
\end{equation}
By choosing the unique normalization of the Haar measure over $A$ such that
\begin{equation}\label{di}
\int_A\langle f(x);y\rangle\operatorname{d}\! y=\delta(x),
\end{equation}
we remark that the following symmetric tempered distributions over $A^2$,
 \begin{equation}
S(x,y)=g(x-y),\quad T(x,y)=\delta(x+y)g( x),
\end{equation}
when interpreted as integral operator kernels, correspond to unitary operators in the pre-Hilbert space of Schwartz--Bruhat test functions over $A$, and thus, through transposition, determine continuous linear automorphisms $S$ and $T$ of the space of tempered distributions over $A$ where the integral kernels of the inverse automorphisms $S^{-1}$ and $T^{-1}$ are given by the complex conjugate kernels of $S$ and $T$ respectively.

The following theorem allows us to interpret the distribution valued solution of the Pachner (3,3)-relation given by \eqref{lcas} with a specific choice of the bi-character as a P-symmetric functional theoretic solution where the vector spaces $X$ and $Y$ are identified with one and the same space of tempered distributions over $A$, the tensor powers $X^{\otimes n}$ being identified with the space of tempered distributions over $A^n$, and the transformation operators $L$, $M$, $R$ are fixed by $L=R=T$ and $M=S$.
\begin{theorem}\label{thm}
 Let $A,f,g, T,S$ be as above and $D$ a tempered distribution over $A^5$ defined  by  \eqref{lcas} where
\begin{equation}
\chi(x,y)={\langle f(x);y\rangle}.
\end{equation}
Then,  the following equalities are satisfied:
\begin{multline}\label{lcasr}
\bar D=(\sigma\otimes T\otimes T^{-1}\otimes T)D=
(T\otimes \sigma \otimes S^{-1}\otimes S)D\\
=(S\otimes S^{-1}\otimes  \sigma\otimes T) D
=(T\otimes T^{-1}\otimes T\otimes   \sigma )D
\end{multline}
where $\bar D$ is the complex conjugate of $D$ and $\sigma$ is the permutation automorphism of the space of tempered distributions over $A^2$.
\end{theorem}
\begin{proof}
Denoting
 \begin{multline}
D^{(1)}\equiv (\sigma\otimes T\otimes T^{-1}\otimes T)D,\quad
D^{(2)}\equiv(T\otimes \sigma \otimes S^{-1}\otimes S)D,\\
D^{(3)}\equiv(S\otimes S^{-1}\otimes  \sigma\otimes T) D,\quad
D^{(4)}\equiv(T\otimes T^{-1}\otimes T\otimes   \sigma )D,
\end{multline}
we check by direct calculation that $D^{(i)}=\bar D$ for each $i\in\{1,2,3,4\}$.
 
The case $i=1$:
\begin{multline}
 D^{(1)}(x,u,y,v,z)=\int_{A^3}T(y,y')\bar T(v,v') T(z,z')D(u,x,y',v',z')\operatorname{d}\! y'\operatorname{d}\! v'\operatorname{d}\! z'\\
 ={g(y)g(z)\over g(v)}D(u,x,-y,-v,-z)={g(y)g(z)\over g(v)\chi(u,z)}\delta(u-x-y)\delta(-y+v-z)\\
={\chi( y,z)\over\chi(u,z)}\delta(u-x-y)\delta(-y+v-z)=\bar\chi( x,z)\delta(u-x-y)\delta(-y+v-z).
\end{multline}
The case $i=2$:
\begin{multline}
 D^{(2)}(x,u,y,v,z)=\int_{A^3}T(x,x')\bar S(v,v') S(z,z')D(x',y,u,v',z')\operatorname{d}\! x'\operatorname{d}\! v'\operatorname{d}\! z'\\
 =\int_{A^2}{g(x)g(z-z')\over g(v-v')}D(-x,y,u,v',z')\operatorname{d}\! v'\operatorname{d}\! z'\\
 =\int_{A^2}{g(x)g(z-z')\over g(v-v')\chi(x,z')}\delta(-x-y+u)\delta(u-v'+z')\operatorname{d}\! v'\operatorname{d}\! z'\\
  =\delta(-x-y+u)\int_{A}{g(x)g(z-z')\over g(v-u-z')\chi(x,z')}\operatorname{d}\! z'
  \end{multline}
and we continue by shifting the integration variable $z'\mapsto z'+z$ and using formula \eqref{di}
\begin{multline}
    =\delta(-x-y+u)\int_{A}{g(x)g(-z')\over g(v-u-z-z')\chi(x,z'+z)}\operatorname{d}\! z'\\
      =\delta(-x-y+u){g(x)\over g(v-u-z)\chi(x,z)}\int_{A}\chi(-x+u+z-v,z')\operatorname{d}\! z'\\
      =\delta(-x-y+u){\bar\chi(x,z)g(x)\over g(v-u-z)}\delta(-x+u+z-v)
      =\delta(-x-y+u)\bar\chi(x,z)\delta(y+z-v).
\end{multline}
The case $i=3$:
\begin{multline}
 D^{(3)}(x,u,y,v,z)=\int_{A^3}S(x,x')\bar S(u,u') T(z,z')D(x',u',v,y,z')\operatorname{d}\! x'\operatorname{d}\! u'\operatorname{d}\! z'\\
 =\int_{A^2}{g(x-x')g(z)\over g(u-u')}D(x',u',v,y,-z) \operatorname{d}\! x'\operatorname{d}\! u'\\
  =\int_{A^2}{g(x-x')g(z)\over g(u-u')\chi(x',z)}\delta(x'-u'+v)\delta(v-y-z) \operatorname{d}\! x'\operatorname{d}\! u'\\
   =\delta(v-y-z)\int_{A}{g(x-x')g(z)\over g(u-v-x')\chi(x',z)} \operatorname{d}\! x'
     \end{multline}
and we continue by shifting the integration variable $x'\mapsto x'+x$ and using formula \eqref{di}
\begin{multline}
   =\delta(v-y-z)\int_{A}{g(-x')g(z)\over g(u-v-x-x')\chi(x'+x,z)} \operatorname{d}\! x'\\
    =\delta(v-y-z){\bar\chi(x,z)g(z)\over g(u-v-x)} \int_{A}\chi(z+u-v-x,-x')\operatorname{d}\! x'\\
     =\delta(v-y-z){\bar\chi(x,z)g(z)\over g(u-v-x)}\delta(z+u-v-x) =\delta(v-y-z)\bar\chi(x,z)\delta(u-y-x).
\end{multline}
The case $i=4$:
\begin{multline}
 D^{(4)}(x,u,y,v,z)=\int_{A^3}T(x,x')\bar T(u,u') T(y,y')D(x',u',y',z,v)\operatorname{d}\! x'\operatorname{d}\! u'\operatorname{d}\! y'\\
 ={g(x)g(y)\over g(u)} D(-x,-u,-y,z,v) ={g(x)g(y)\over g(u)\chi(x,v)} \delta(-x+u-y)\delta(-y-z+v)\\
 ={\chi(x,y)\over \chi(x,v)} \delta(-x+u-y)\delta(-y-z+v) =\bar\chi(x,z) \delta(-x+u-y)\delta(-y-z+v).
\end{multline}
\end{proof}

\def\cprime{$'$} \def\cprime{$'$}

  \end{document}